\documentclass[copyright]{eptcs}

\providecommand{\event}{DCFS 2010}
\usepackage{breakurl}

\usepackage{graphicx} 
\usepackage{amsmath,amssymb}

\newcommand{\inititem}{\setcounter{enumi}{0}\nextitem}
\newcommand{\nextitem}{\addtocounter{enumi}{1}(\theenumi)}

\newcommand{\nfa}{\textrm{NFA}}
\newcommand{\dfa}{\textrm{DFA}}

\newcommand*{\qed}{\raisebox{0.5ex}[0ex][0ex]{\framebox[1ex][l]{}}}
\newtheorem{theorem}{Theorem}
\newtheorem{lemma}[theorem]{Lemma}
\newtheorem{corollary}[theorem]{Corollary}
\newenvironment{proof}{%
  \par\noindent
  {\rmfamily\itshape\mdseries Proof\/}:\hspace{\labelsep}\ignorespaces}%
  {\mbox{}\nolinebreak\hfill~%
  {\qed}
  \medbreak
}

\begin{document}

\title{The Magic Number Problem for\\ Subregular Language Families}

\def\titlerunning{Magic Numbers for Subregular Language Families}
\def\authorrunning{M.~Holzer, S.~Jakobi, M.~Kutrib}

\author{%
Markus Holzer%
\qquad
Sebastian Jakobi%
\qquad
Martin Kutrib
\institute{Institut f\"ur Informatik, Universit\"at Giessen,\\
  Arndtstr.~2, 35392 Giessen, Germany\\
  email: \texttt{$\{$holzer,jakobi,kutrib$\}$@informatik.uni-giessen.de}}
}

\maketitle

\begin{abstract}
  We investigate the magic number problem, that is, the question
  whether there exists a minimal $n$-state nondeterministic finite
  automaton ($\nfa$) whose equivalent minimal deterministic finite
  automaton ($\dfa$) has~$\alpha$ states, for all~$n$ and~$\alpha$
  satisfying $n\leq \alpha\leq 2^n$. A number~$\alpha$ not satisfying
  this condition is called a \emph{magic number} (for~$n$).  It was
  shown in \cite{jiraskova:2007:dbunfafa} that no magic numbers exist
  for general regular languages, while in~\cite{Ge07} trivial and
  non-trivial magic numbers for unary regular languages were
  identified.  We obtain similar results for automata accepting
  subregular languages like, for example, combinational languages,
  star-free, prefix-, suffix-, and infix-closed languages, and
  pre\-fix-, suffix-, and infix-free languages, showing that there are
  only trivial magic numbers, when they exist. For finite languages we
  obtain some partial results showing that certain numbers are
  non-magic.
\end{abstract}

\section{Introduction}
\label{sec:intro}

Nondeterministic finite automata ($\nfa$s) are probably best known for
being equivalent to right-linear context-free grammars and, thus, for
capturing the lowest level of the Chomsky-hierarchy, the family of
regular languages.  It is well known that $\nfa$s can offer
exponential saving in space compared with deterministic finite
automata ($\dfa$s), that is, given some $n$-state $\nfa$ one can
always construct a language equivalent $\dfa$ with at most~$2^n$
states~\cite{RaSc59}. This so-called \emph{powerset construction}
turned out to be optimal, in general. That is, the bound on the number
of states is tight in the sense that for an arbitrary~$n$ there is
always some $n$-state $\nfa$ which cannot be simulated by any $\dfa$
with less than~$2^n$
states~\cite{lupanov:1963:acottfa,Meyer:1971:edagfs,moore:1971:bssspe}. On
the other hand, there are cases where nondeterminism does not help for
the succinct representation of a language compared to $\dfa$s.  These
two milestones from the early days of automata theory form part of an
extensive list of equally striking problems of $\nfa$ related
problems, and are a basis of descriptional complexity. Moreover, they
initiated the study of the power of resources and features given to
finite automata. For recent surveys on descriptional complexity issues
of regular languages we refer to, for
example,~\cite{HoKu09,HoKu09a,HoKu10}.

Nearly a decade ago a very fundamental question on the well known
subset construction was raised in~\cite{IKT00}: Does there always
exists a minimal $n$-state $\nfa$ whose equivalent minimal $\dfa$
has~$\alpha$ states, for all~$n$ and~$\alpha$ with $n\leq \alpha\leq
2^n$? \mbox{A number}~$\alpha$ not satisfying this condition is called a
\emph{magic number} for~$n$. The answer to this simple question turned
out not to be so easy. For $\nfa$s over a two-letter alphabet it was
shown that $\alpha=2^n-2^k$ or $2^n-2^k-1$, for $0\leq k\leq
n/2-2$~\cite{IKT00}, and $\alpha=2^n-k$, for $5\leq k\leq 2n-2$ and
some coprimality condition for~$k$~\cite{IMP03}, are non-magic.
In~\cite{Ji01} it was proven that the integer~$\alpha$ is non-magic,
if $n\leq\alpha\leq 1+n(n+1)/2$. This result was improved by showing
that $\alpha$ is non-magic for $n\leq\alpha\leq 2^{\sqrt[3]{n}}$
in~\cite{Ji08}.  Further non-magic numbers for two-letter input
alphabet were identified in~\cite{Ge05} and~\cite{MaSa08}. It turned
out that the problem becomes easier if one allows more input
letters. In fact, for exponentially growing alphabets there are no
magic numbers at all~\cite{Ji01}. This result was improved to less
growing alphabets in~\cite{Ge05}, to constant alphabets of size four
in~\cite{jiraskova:2007:dbunfafa}, and very recently to three-letter
alphabets~\cite{Ji09}.  Magic numbers for unary $\nfa$s were recently
studied in~\cite{Ge07} by revising the Chrobak normal-form for
$\nfa$s. In the same paper also a brief historical summary of the
magic number problem can be found.  Further results on the magic
number problem (in particular in relation to the operation problem on
regular languages) can be found, for example, in~\cite{Ji08,Ji09a}.

To our knowledge the magic number problem was not systematically
studied for subregular languages families (except for unary
languages). Several of these subfamilies are well motivated by their
representations as finite automata or regular expressions: finite
languages (are accepted by acyclic finite automata), combinational
languages (are accepted by automata modeling combinational circuits),
star-free languages or regular non-counting languages (which can be
described by regular-like expression using only union, concatenation,
and complement), prefix-closed languages (are accepted by automata
where all states are accepting), suffix-closed (or multiple-entry or
fully-initial) languages (are accepted by automata where the
computation can start in any state), infix-closed languages (are
accepted by automata where all states are both initial and accepting),
suffix-free languages (are accepted by non-returning automata, that
is, automata where the initial state does not have any in-transition),
prefix-free languages (are accepted by non-exiting automata, that is,
automata where all out-transitions of every accepting state go to a
rejecting sink state), and infix-free languages (are accepted by
non-returning and non-exiting automata, where these conditions are
necessary, but not sufficient).

The hierarchy of these and some further subregular language families
is well known.  We study all families mentioned with respect to the
magic number problem, and show---except for finite languages, where
only some partial results will be presented---that there are only
trivial magic numbers, whenever they exist.

\section{Definitions}
\label{sec:defs}

Let $\Sigma^*$ denote the set of all \emph{words} over the finite
alphabet $\Sigma$.  For $n\geq 0$ we write~$\Sigma^n$ for the set of
all words of length~$n$.  The \emph{empty word} is denoted by
$\lambda$ and $\Sigma^+ = \Sigma^* \setminus \{\lambda\}$. A
\emph{language}~$L$ over~$\Sigma$ is a subset of~$\Sigma^*$.  For the
length of a word~$w$ we write $|w|$. Set inclusion is denoted by
$\subseteq$ and strict set inclusion by $\subset$.  We write $2^{S}$
for the power set and $|S|$ for the cardinality of a set $S$.

A \emph{nondeterministic finite automaton} ($\nfa$) is a quintuple
$A=(Q,\Sigma,\delta,q_0,F)$, where $Q$ is the finite set of
\emph{states}, $\Sigma$ is the finite set of \emph{input symbols},
$q_0\in Q$ is the \emph{initial state}, $F\subseteq Q$ is the set of
\emph{accepting states}, and \mbox{$\delta: Q\times\Sigma\rightarrow
  2^Q$} is the \emph{transition function}.  As usual the transition
function is extended to \mbox{$\delta: Q\times \Sigma^*\to 2^Q$}
reflecting sequences of inputs: \mbox{$\delta(q,\lambda) =\{q\}$} and
$\delta(q, aw) = \bigcup_{q'\in\delta(q,a)} \delta(q',w)$, for $q\in
Q$, $a\in \Sigma$, and $w\in \Sigma^*$.  A word $w\in \Sigma^*$ is
\emph{accepted} by $A$ if $\delta(q_0,w)\cap F\ne\emptyset$.  The
\emph{language accepted} by~$A$ is $L(A) =\{\,w\in
\Sigma^*\mid\mbox{$w$ is accepted by $A$}\,\}$.

A finite automaton is \emph{deterministic} ($\dfa$) if and only if
$|\delta(q,a)|=1$, for all $q\in Q$ and $a\in\Sigma$. In this case we
simply write $\delta(q,a)=p$ for $\delta(q,a)=\{p\}$ assuming that the
transition function is a mapping $\delta:Q\times\Sigma\rightarrow
Q$. So, any $\dfa$ is complete, that is, the transition function is
total, whereas for $\nfa$s it is possible that $\delta$ maps to the
empty set.  Note that a sink state is counted for $\dfa$s, since they
are always complete, whereas it is not counted for $\nfa$s, since
their transition function may map to the empty set.  In the sequel we
refer to the $\dfa$ obtained from an $\nfa$
$A=(Q,\Sigma,\delta,q_0,F)$ by the power-set construction as
$A'=(2^Q,\Sigma,\delta',\{q_0\},F')$, where
$\delta'(P,a)=\bigcup_{p\in P} \delta(p,a)$, for $P\subseteq Q$ and
$a\in\Sigma$, and $F'=\{\,P\subseteq Q\mid P\cap F\neq\emptyset\,\}$.

As already mentioned in the introduction,
in~\cite{jiraskova:2007:dbunfafa} it was shown that for all
integers~$n$ and~$\alpha$ such that $n\leq\alpha\leq 2^n$, there
exists an $n$-state nondeterministic finite automaton~$A_{n,\alpha}$
whose equivalent minimal deterministic finite automaton has
exactly~$\alpha$ states. Since some of our constructions rely on this
proof and for the sake of completeness and readability we briefly
recall the sketch of the construction.  In the following we call the
\nfa~$A_{n,\alpha}$ the Jir{\'a}sek-Jir{\'a}skov{\'a}-Szabari
automaton, or for short the JJS-automaton.

\begin{theorem}[\cite{jiraskova:2007:dbunfafa}]
  For all integers~$n$ and~$\alpha$ such that $n\leq\alpha\leq 2^n$,
  there exists an $n$-state nondeterministic finite automaton~$A_{n,\alpha}$
  whose equivalent minimal deterministic finite automaton has
  exactly~$\alpha$ states.
\end{theorem}

In the construction for some fixed integer~$n$ the cases $\alpha=n$
and $\alpha=2^n$ are treated separately by appropriate witness
languages.  For the remaining cases it is first shown that
every~$\alpha$ satisfying $n<\alpha<2^n$ can be written as a specific
sum of powers of two. In particular, for all integers~$n$ and~$\alpha$
such that $n<\alpha<2^n$, there exist integers~$k$ and~$m$ with $1\leq
k\leq n-1$ and $1\leq m<2^k$, such that
\begin{eqnarray*}
  \alpha &=& n-(k+1)+2^k+m\label{eqn:alpha}\\
\noalign{\hbox{and}}
 m &=& (2^{k_1}-1)+(2^{k_2}-1)+\cdots+(2^{k_{\ell-1}}-1)+
   \begin{cases}
     (2^{k_\ell}-1) & \\ 
     2\cdot (2^{k_\ell}-1) & 
   \end{cases}\label{eqn:terms-of-m}
\end{eqnarray*}
where $1\leq\ell\leq k-1$ and $k\geq k_1 > k_2 >\cdots >k_\ell\geq 1$.
Then \nfa s are constructed such that the powerset construction yields
\dfa s whose number of states is exactly one of these powers of two,
which finally have to be combined appropriately to lead to a single
$n$-state $\nfa$ $A_{n,\alpha}$ whose equivalent minimal \dfa\ has
exactly~$\alpha$ states. Automaton $A_{n,\alpha}$ is depicted in
Figure~\ref{fig:Galinas-automaton-Cnk}, where the following
$d$-transitions are not shown:
$$
\delta(i,d)=
\begin{cases}
 \{0,2,3,4,\ldots,k-k_i+1\} & \mbox{if $1\leq i\leq\ell-1$}\\
 \{0,1,\ldots,k-k_i+1\} & \mbox{if $i=\ell$ and~$m$ is of the first form}\\
 \{0,2,3,4,\ldots,k-k_i+1\} & \mbox{if $i=\ell$ and~$m$ is of the second form}\\
 \{0,1,\ldots,k-k_i+1\} & \mbox{if $i=\ell+1$ and~$m$ is of the second form}\\
 \emptyset & \mbox{otherwise.} 
\end{cases}
$$
\begin{figure}[ht]
  \centering 
\includegraphics[width=0.9\textwidth]{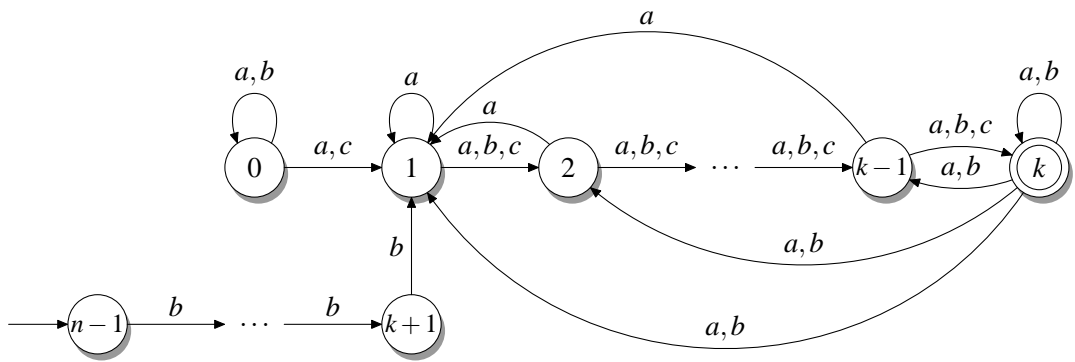}
\caption{Jir{\'a}sek-Jir{\'a}skov{\'a}-Szabari's (JJS)
  nondeterministic finite automaton~$A_{n,\alpha}$ with~$n$ states
  (\mbox{$d$-transitions} are not shown) accepting a language for which the
  equivalent minimal \dfa\ needs exactly $\alpha=n-(k+1)+m$ states.}
  \label{fig:Galinas-automaton-Cnk}
\end{figure}

\section{Results}
\label{sec:results}

We systematically investigate the magic number problem for the
aforementioned subregular language families.  For the remaining
theorems of this paper, when speaking of an $n$-state \nfa\ we always
mean a minimal $\nfa$.  Given a subregular language family, if $f(n)$
is the number of states that is sufficient and necessary in the worst
case for a $\dfa$ to accept the language of an $n$-state $\nfa$
belonging to the family, then a number $\alpha$ with $f(n) < \alpha
\leq 2^n$ is called a \emph{trivial} magic number. Similarly, if
$g(n)$ is the number of states that is necessary for any $\dfa$
simulating an arbitrary $n$-state $\nfa$, then all numbers $\alpha$
with $\alpha < g(n)$ is also called a \emph{trivial} magic number. For
example, for infix-free languages $g(n)$ is shown to be $n+1$ in
Theorem~\ref{thm:prefix-suffix-infix-free}, while $f(n)$ is known to
be $2^{n-2}+2$~\cite{BHK09a}.  Due to space constraints most proofs
are omitted.

An observation from~\cite{BHK09a} shows that the magic number problem
for \emph{elementary} and \emph{combinational} languages is trivial.

\subsection{Star-Free Languages and Power Separating Languages}
\label{sec:star-free}

A language $L\subseteq\Sigma^*$ is \emph{star-free} (or regular
\emph{non-counting}) if and only if it can be obtained from the
elementary languages~$\{a\}$, for $a\in \Sigma$, by applying the
Boolean operations union, complementation, and concatenation finitely
often.  These languages are exhaustively studied, for example,
in~\cite{McNaPa71}.  Since regular languages are closed under Boolean
operations and concatenation, every star-free language is regular. On
the other hand, not every regular language is star free.

Here we use an alternative characterization of star-free languages by
so called permutation-free automata~\cite{McNaPa71}: A regular
language~$L\subseteq\Sigma^*$ is star-free if and only if the minimal
\dfa\ accepting~$L$ is \emph{permutation-free}, that is, there is
\emph{no} word $w\in\Sigma^*$ that induces a non-trivial permutation
of any subset of the set of states. Here a trivial permutation is
simply the identity permutation. Observe that a word~$uw$ induces a
non-trivial permutation $\{q_1,q_2,\ldots,q_n\}\subseteq Q$ in a \dfa\
with state set~$Q$ and transition function~$\delta$ if and only
if~$wu$ induces a non-trivial permutation
$\{\delta(q_1,u),\delta(q_2,u),\ldots,\delta(q_n,u)\}$ in the same
automaton.  Further, if one finds a non-trivial permutation consisting
of multiple disjoint cycles, it suffices to consider a single cycle.
Before we show that no magic numbers exist for star-free languages we
prove a useful lemma on permutations in (minimal) \dfa s obtained by
the powerset construction.

\begin{lemma}\label{lem:non-trivial-permuation-in-powerset-automaton}
  Let~$A$ be a nondeterministic finite automaton with state set~$Q$
  over alphabet~$\Sigma$, and assume that~$A'$ is the equivalent
  minimal deterministic finite automaton, which is
  non-permutation-free. If the word~$w$ in~$\Sigma^*$ induces a
  non-trivial permutation on the state set $\{P_0,P_1,\ldots,
  P_{n-1}\}\subseteq 2^Q$ of~$A'$, that is, $\delta'(P_i,w)=P_{i+1}$,
  for $0\leq i<n-1$, and $\delta'(P_{n-1},w)=P_0$, then there are no
  two states~$P_i$ and $P_j$ with $i\neq j$ such that $P_i\subseteq
  P_j$.
\end{lemma}

\begin{proof}
  Assume to the contrary that $P_0\subseteq P_i$ (possibly after a
  cyclic shift), for some $0<i\leq n-1$. Then one can show by
  induction that $\delta'(P_0,v)\subseteq\delta'(P_i,v)$, for every
  word $v\in\Sigma^*$. In particular, this also holds true for the
  word~$w$ that induces the non-trivial permutation on the state set
  $\{P_0,P_1,\ldots, P_{n-1}\}$. But then $ P_{ki\bmod
    n}=\delta'(P_0,w^{ki})\subseteq\delta'(P_i,w^{ki})=P_{(k+1)i\bmod
    n}$, for $k\geq 0$, and one finds the chain of inclusions
  $P_0\subseteq P_i\subseteq P_{2i\bmod n}\subseteq P_{3i\bmod
    n}\subseteq\cdots\subseteq P_{ni\bmod n}=P_0$, which implies
  $P_0=P_i$, a contradiction.
\end{proof}

Now we are prepared for the main theorem, which
utilizes Lemma~\ref{lem:non-trivial-permuation-in-powerset-automaton}.

\begin{theorem}\label{thm:star-free}
  For all integers~$n$ and~$\alpha$ such that $n\leq\alpha\leq 2^n$,
  there exists an $n$-state nondeterministic finite automaton
  accepting a \emph{star-free} language whose equivalent minimal
  deterministic finite automaton has exactly~$\alpha$ states.
\end{theorem}

The previous theorem generalizes to all language families that are a
superset of the family of star-free languages such as, for example, the
family of power separating languages introduced
in~\cite{shyr:1974:psrl}. 

\subsection{Stars and  Comet Languages}
\label{sec:stars-comets-etc}

A language $L\subseteq\Sigma^*$ is a \emph{star language} if and only
if $L=H^*$, for some regular language
$H\subseteq\Sigma^*$, and $L\subseteq\Sigma^*$ is a \emph{comet
  language} if and only if it can be represented as concatenation
$G^*H$ of a regular star language~$G^*\subseteq\Sigma^*$ and a regular
language $H\subseteq\Sigma^*$, such that $G\ne\{\lambda\}$ and
$G\ne\emptyset$. Star languages and comet languages were
introduced in~\cite{Br67} and~\cite{BrCo67}. Next, a language
$L\subseteq\Sigma^*$ is a \emph{two-sided comet language} if and only
if $L=EG^*H$, for a regular star language $G^*\subseteq\Sigma^*$ and
regular languages $E,H\subseteq\Sigma^*$, such that $G\ne\{\lambda\}$ and
$G\ne\emptyset$. So, (two-sided) comet languages are always infinite.
Clearly, every star language not equal to $\{\lambda\}$
is also a comet language and every comet is a two-sided comet
language, but the converse is not true in general. 

\begin{theorem}\label{thm:star}
  For all integers~$n$ and~$\alpha$ such that $n\leq\alpha\leq 2^n$,
  there exists an $n$-state nondeterministic finite automaton
  accepting a \emph{star} language whose equivalent minimal
  deterministic finite automaton has exactly~$\alpha$ states.  The
  statement remains valid for \emph{(two-sided) comet} languages.
\end{theorem}

\subsection{Subword Specific Languages}

In this section we consider languages for which for every word in the
language either all or none of its prefixes, suffixes or infixes
belong to the same language.  Again, there are only trivial magic
numbers. We start with subword-free languages.

A language $L\subseteq\Sigma^*$ is \emph{prefix-free} if and only if 
$y\in L$ implies $yz\notin L$, for all $z\in\Sigma^+$,
\emph{infix-free} if and only if 
$y\in L$ implies $xyz\notin L$, for all $xz\in\Sigma^+$, and
\emph{suffix-free} if and only if 
$y\in L$ implies $xy\notin L$, for all $x\in\Sigma^+$.

\begin{theorem}\label{thm:prefix-suffix-infix-free}
  Let~$A$ be a minimal $n$-state \nfa\ accepting a non-empty prefix-,
  suffix- or infix-free language. Then any equivalent minimal \dfa\
  accepting language~$L(A)$ needs at least $n+1$ states.
\end{theorem}

In the following we show that no non-trivial magic numbers exist for
subword-free languages.  The upper bound for the deterministic blow-up
in prefix- and suffix-free languages is $2^{n-1}+1$ and for infix-free
languages it is $2^{n-2}+2$, so all numbers above are trivially magic.

\begin{theorem}\label{thm:prefix-suffix-free}
  For all integers~$n$ and~$\alpha$ such that $n< \alpha\leq
  2^{n-1}+1$, there exists an $n$-state nondeterministic finite
  automaton accepting a \emph{prefix-free} language whose equivalent
  minimal deterministic finite automaton has exactly~$\alpha$
  states. The statement remains true for \nfa s accepting
  \emph{suffix-free} languages.
\end{theorem}

For infix-free regular languages the situation is slightly different
compared to above.

\begin{theorem}\label{thm:infix-free}
  For all integers~$n$ and~$\alpha$ such that $n< \alpha\leq 2^{n-2}+2$,
  there exists an $n$-state nondeterministic finite automaton
  accepting an \emph{infix-free} language
  whose equivalent minimal deterministic finite automaton has
  exactly~$\alpha$ states. 
\end{theorem}

Next, we consider prefix-, infix-, and suffix-closed languages.  A
language $L\in\Sigma^*$ is \emph{prefix-closed} if and only if $xy\in
L$ implies $x\in L$, for $x\in\Sigma^*$, \emph{infix-closed} if and
only if $xyz\in L$ implies $y\in L$, for $x,z\in\Sigma^*$, and
\emph{suffix-closed} if and only if $yz\in L$ implies $z\in L$, for
$z\in\Sigma^*$.  We use the following results
from~\cite{kao:2007:nfaasfib}.

\begin{theorem}\label{thm:kao2007}
\inititem\ 
A nonempty regular language is prefix-closed if and only if it
    is accepted by some nondeterministic finite automaton with all
    states accepting.
\nextitem\
A nonempty regular language is infix-closed if and only if it
    is accepted by some nondeterministic finite automaton with multiple
    initial states with all states both initial and accepting.
\end{theorem}

Prefix-closed languages reach the upper bound of $2^n$ states, and 
for infix-closed languages it is \mbox{$2^{n-1}+1$}. Up to these bounds
the only magic number for both language families is $n$
(except for $n=1$). The upper bound for suffix-closed
languages is $2^{n-1}+1$, and up to this, no magic numbers exist.

\begin{theorem}\label{thm:infix-closed}
  For all integers~$n$ and~$\alpha$ such that $n< \alpha\leq 2^{n-1}+1$,
  there exists an $n$-state nondeterministic finite automaton
  accepting an \emph{infix-closed} language
  whose equivalent minimal deterministic finite automaton has
  exactly~$\alpha$ states. The case $\alpha=n$ can only be reached for
  $n=1$. 
\end{theorem}

\begin{proof}
  For the second statement, note that each \dfa\ accepting a language
  $L\neq\Sigma^*$ needs a non-accepting state, which the minimal \nfa\
  cannot have, due to Theorem~\ref{thm:kao2007}. So, $\Sigma^*$ is the
  only infix-closed language, for which the size of the minimal \dfa\
  equals the size of an equivalent minimal \nfa. Both have a single
  state.  The case $\alpha=2^{n-1}+1$ is discussed in~\cite{BHK09a}.
  For the remaining, assume \mbox{$n<\alpha\leq 2^{n-1}$.}  In this
  case, the JJS-automaton $A_{n,\alpha}=(Q,\Sigma,\delta,n-1,\{k\})$
  has a non-empty initial tail of states, that is, the initial state
  is equal to state~$n-1$.  {F}rom $A_{n,\alpha}$ we construct an
  automaton $A_1=(Q,\Sigma\cup\{\#,\$\},\delta_1,Q,Q)$ with all states
  initial and accepting and transition function
  $\delta_1(k,\#)=\{k\}$, $\delta_1(q,\$)=\{q-1\}$ if \mbox{$k+2\leq q\leq
  n-1$}, $\delta_1(k+1,\$)=\{1\}$ and $\delta_1(q,a)=\delta(q,a)$ for
  $0\leq q\leq k$ and $a\in\Sigma$.  This \nfa\ with multiple initial states
  can be converted into an equivalent \nfa\ $A_2$ with initial state
  $n-1$ and the transition function $\delta_2(n-1,a)=\bigcup_{q\in Q}
  \delta_1(q,a)$ and $\delta_2(q,a)=\delta_1(q,a)$ for all
  $a\in\Sigma\cup\{\#,\$\}$ and $q\in Q\setminus\{n-1\}$.

	With $S_1=\{\,(\$^i,\$^{n-(k+1)-i}c^{k-1}) \mid 0 \leq i \leq
  n-(k+1)\,\}$, $S_2=\{\,(\$^{n-(k+1)}c^i,c^{k-1-i}) \mid 1 \leq i
  \leq k-1\,\}$ and $S_3=\{(\$^{n-(k+1)}d,c^{k})\}$, one can easily
  check, that $S=S_1\cup S_2\cup S_3$ is a fooling set for $L(A_2)$:
  Different pairs from~$S_1$ result in a word beginning with more than
  $n-(k+1)$ $\$$-symbols, pairs from $S_2$ result in too many
  $c$-symbols,~$c^k$ from $S_3$ cannot be combined with any other word
  and mixing pairs from~$S_1$ and~$S_2$ either results in a word
  containing the infix $\$c^i \$$ or, if $(\$^{n-(k+1)}, c^{k-1})$ is
  chosen from $S_1$, in $\$^{n-(k+1)} c^{i+k-1}$, which ends with too
  many $c$-symbols.

  In the corresponding powerset automaton $A_2'$, by reading prefixes
  of $\$^{n-(k+1)}$, one reaches \mbox{$n-(k+1)$} states $\{n-1\}$,
  $\{n-2,\dots,k+1,1\}$, \dots, $\{k+1,1\}$. After reading
  $\$^{n-(k+1)}$, $A_2'$ is in state $\{1\}$ and from there, according
  to the JJS-construction, $2^k + m$ states from $2^{\{0,1,\dots,k\}}$
  are reachable. So we have exactly~$\alpha$ states.  To see that no
  further states can be reached, note that the transition function
  differs from the one of the JJS-automaton only in states $k+1,\dots,
  n-1$ and state $k$.  The $\#$-transition in state $k$ gives no new
  reachable states and reading~$\$$ always leads to either a state
  $\{n-i,\dots,k+1,1\}$, for some $1\leq i\leq n-k+1$, or to state
  $\{1\}$ or the empty set.  So, the only interesting transitions are
  those of the initial state $\{n-1\}$ on the input symbols $a$, $b$,
  $c$ and $d$. Reading~$a$ or~$b$ leads to $\{0,\dots, k\}$, reading
  $c$ to $\{1,\dots,k\}$ and on input~$d$,~$A_2'$ enters the state
  $\delta(q,d)$ for the largest $q\in Q$ for which this transition is
  defined. All these states were already counted.

  To prove that any two distinct states $M,N\subseteq
  Q\setminus\{n-1\}$ are pairwise inequivalent, without loss of
  generality, pick an element $q\in M\setminus N$.  If $q \leq k$, the
  word $c^{k-q}\#$ distinguishes $M$ and $N$. Otherwise, if~$q\geq
  k+1$, one can drive it to state $1$ by reading $\$$-symbols, and
  then $c^{k-1}$ distinguishes the two states.

  Finally, state $\{n-1\}$ is inequivalent with any state $N\subseteq
  Q\setminus\{n-1\}$ by the input word $\$^{n-(k+1)}c^{k-1}$.
\end{proof}

The family of infix-closed languages is a subset of the family of
suffix-closed languages, so the previous theorem generalizes to the latter
language family, except for $n$ which is not magic for $n\geq 1$ anymore:

\begin{corollary}\label{cor:suffix-closed}
  For all integers~$n$ and~$\alpha$ such that $n\leq \alpha\leq
  2^{n-1}+1$, there exists an $n$-state nondeterministic finite
  automaton accepting a \emph{suffix-closed} language whose equivalent
  minimal deterministic finite automaton has exactly~$\alpha$ states.
\end{corollary}

Since the upper bound for the deterministic blow-up of prefix-closed
languages is greater than that of infix-closed languages, we need to
treat them separately here.
  
\begin{theorem}\label{thm:prefix-closed}
  For all integers~$n$ and~$\alpha$ such that $n< \alpha\leq 2^n$,
  there exists an $n$-state nondeterministic finite automaton
  accepting a \emph{prefix-closed} language
  whose equivalent minimal deterministic finite automaton has
  exactly~$\alpha$ states. The case $\alpha=n$ can only be reached for
  $n=1$.
\end{theorem}

\subsection{Finite Languages}
\label{sec:finite}

For finite languages the magic number problem turns out to be more
challenging which seems to coincide with the fact, that the upper
bounds for the deterministic blow-up of finite languages differ much
from these of infinite language families.  In~\cite{SaYu97a} it was
shown that for each $n$-state \nfa\ over an alphabet of size $k$,
there is an equivalent \dfa\ with at most $O(k^{n/(\log(k)+1)})$
states.  This matches an earlier result of~$O(2^{n/2})$ for finite
languages over binary alphabets~\cite{Ma73}.

In this section we give some partial results for finite languages over
a binary alphabet, that is, we show that a roughly quadratic interval
beginning at $n+1$ contains only non-magic numbers and that numbers of
some exponential form $2^{(n-1)/2}+2^i$ are non-magic, too.  Note that
for finite languages, $n$ is a trivial magic number, since any \dfa\
needs a non-accepting sink state which is not necessary for an~$\nfa$.

\begin{theorem}
\begin{sloppypar}
  For all integers~$n$ and~$\alpha$ such that $n+1\leq \alpha\leq
  (\frac{n}2)^2 + \frac{n}2+1$ if~$n$ is even, and \mbox{$n+1\leq\alpha\leq
  (\frac{n-1}2)^2 + n +1$} if $n$ is odd, there exists an $n$-state
  nondeterministic finite automaton accepting a \emph{finite} language
  over a binary alphabet whose equivalent minimal deterministic finite
  automaton has exactly~$\alpha$ states.
\end{sloppypar}
\end{theorem}

\begin{proof}
  The case $\alpha=n+1$ can be seen with the witness language
  $\{a,b\}^n$. So, assume $n+1<\alpha$.  Then there exist integers $k$
  and $m$ such that
$$k = \max\{\, x\geq 0 \mid \alpha > 1 +\sum_{i=0}^x
n-2i\}\quad\mbox{and}\quad m = \alpha - 1 - \sum_{i=0}^k n-2i. $$ Let
$A=(\{1,\dots,n\},\{a,b\},\delta,1,\{n\})$ be an \nfa\ with
$\delta(q,a) =\{q+1,2q+1,2q+2,\dots,n\}$, \mbox{$1\leq q \leq k$},
$\delta(k+1,a) = \{(k+1)+1,n-(m-1),n-(m-1)+1,\dots,n\}$,
$\delta(q,a) =\{q+1\}$, for $k+1< q < n$, and $\delta(q,b) = \{q+1\}$, for $q < n$.

The transitions on $b$ ensure the minimality of $A$ and the
inequivalence of states in the corresponding powerset automaton $A'$.
To count all reachable states of $A'$, we partition the set
$\{a,b\}^*$ as follows:
\[\{a,b\}^* = \bigcup_{i=0}^{k} \{b^ia\}\{a,b\}^* \cup
\{b^{k+1}\}\{a,b\}^* \cup \{b\}^*.\] With words from $\{b\}^*$, the
singletons $\{1\},\dots,\{n\}$ and $\emptyset$ are reachable---which
gives $n+1$ states.  Next, let $w=b^i a w'$ and $w'\in\{a,b\}^j$ for
some integers $0\leq i\leq k-1$ and $j\geq k$. Then
  \begin{align}
  \delta'(\{1\},w) &=  \delta'(\{i+1\},aw') \nonumber \\
  &=  \delta'(\{i+2, 2(i+1)+1, 2(i+1)+2,\dots,n\},w') \nonumber \\
  &=  \{i+j+2, 2(i+1)+j+1, 2(i+1)+j+2,\dots,n\}. \label{finite:delta_w_1}
  \end{align} 
  Since we already counted the singleton sets and the empty set, we
  have to count sets of the form~\eqref{finite:delta_w_1} having at
  least two elements. We conclude that the set $\{i+j+2, 2(i+1)+j+1,
  2(i+1)+j+2,\dots,n\}$ has cardinality at least~$2$ if and only if we
  have $\{2(i+1)+j+1, 2(i+1)+j+2,\dots,n\} \neq \emptyset$, which in
  turn holds if and only if $0\leq j \leq n-2(i+1) -1$.
  So, there are $n-2(i+1)$ states reachable for a fixed $i\leq k-1$,
  which gives $\sum_{i=1}^k n-2i$ states that are reachable by reading
  words from $\{b^ia\}\{a,b\}^j$ with $0\leq i\leq k-1$.

  Now let $w=b^k a w'$ for some $w'\in\{a,b\}^j$. Then
  \begin{align}
  \delta'(\{1\},w) &=  \delta'(\{k+1\},aw') \nonumber\\
  &=  \delta'(\{k+2, n-(m-1), n-(m-1)+1,\dots,n\},w')\nonumber \\
  &=  \{k+j+2, n-(m-1)+j, n-(m-1)+j+1,\dots,n\}. \label{finite:delta_w_2}
  \end{align}
  These sets contain at least two elements if and only if $0\leq j\leq
  m-1$. Therefore, exactly~$m$ sets of the
  form~\eqref{finite:delta_w_2} are reachable in $A'$.  Summing up, we
  get
  $$m+n+1+\sum_{i=1}^k n-2i= \alpha-1 -\sum_{i=0}^k n-2i + 1
  +\sum_{i=0}^k n-2i = \alpha$$ states.

  To see that we have not multiply counted any state, note that there
  is no reachable set of states that satisfies
  \eqref{finite:delta_w_1} and \eqref{finite:delta_w_2}: Assume for
  some integers $i,j,j',k$ with $i<k$ that
  \[\{i+j+2,2(i+1)+j+1,\dots,n\}=\{k+j'+2,n-(m-1)+j',\dots,n\}.\]
  Then of course $i+j=k+j'$, and by definition of $k$ and $m$ we have
  $1\leq m \leq n-2(k+1) -1$. Finally, we derive
  \begin{multline*}
  n-(m-1)+j' \geq n- n+2(k+1) +1 +j'
  = k + k + j' +4\\
  > i + k + j' +4
  = i + i + j +4
  > 2(i+1) +j +1
  \end{multline*}
  which is a contradiction to the assumption.
\end{proof}

For our last theorem we use the following results presented
in~\cite{Ma73}: For an integer $n$ let $k=\lceil\frac{n}2\rceil$ and
$A_n=(\{1,\dots,n\},\{a,b\},\delta,1,\{n\})$ be an \nfa\ with
transitions $\delta(q,a)=\{q+1,k+1\}$ if $q<k$, \mbox{$\delta(q,a)=\{q+1\}$}
if $k\leq q <n$, and $\delta(q,b)=\{q+1\}$ if $q<n$ and $q\neq k$.
Then in~\cite{Ma73} it is shown that $A_n$ is minimal and that the
minimal equivalent \dfa\ has exactly $2^{(n/2)+1}-1$ states if $n$ is
even, and~$3\cdot 2^{(n+1)/2 - 1} -1$ states if $n$ is odd.  This
minimal \dfa\ spans a binary tree on inputs $a$ and $b$ as depicted in
Figure~\ref{fig:mandl-nfa}.
\begin{figure}[!hbt]
  \centering
  \includegraphics[scale=0.80]{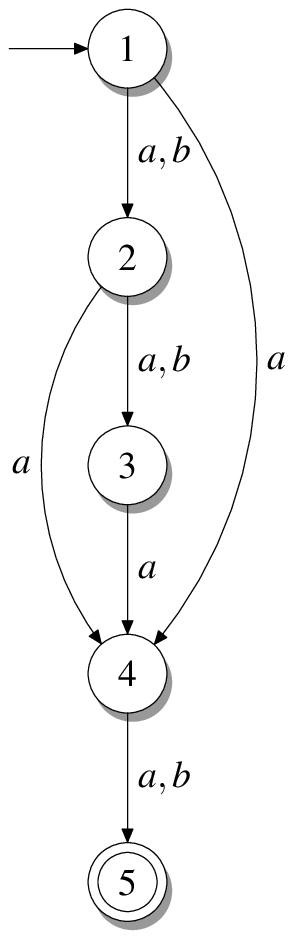}\includegraphics[scale=0.80]{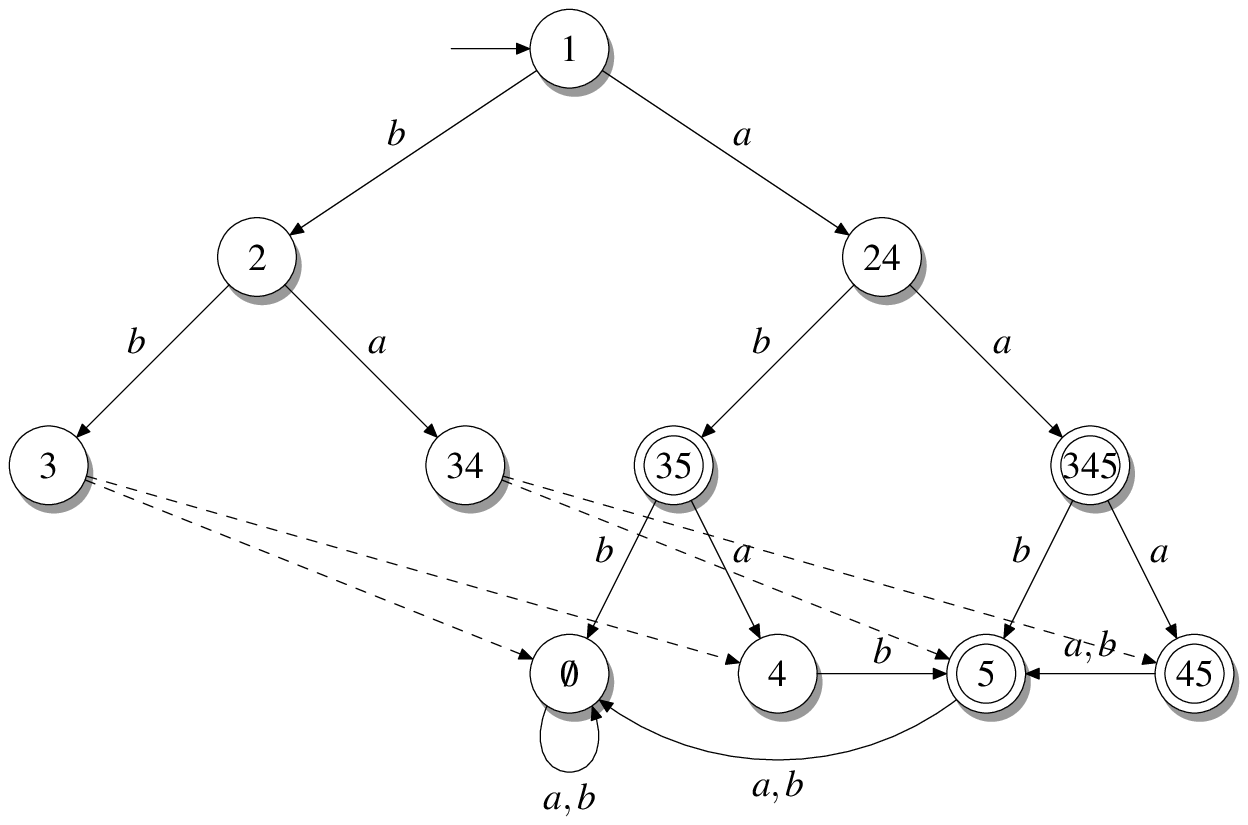}
  \caption{The $\nfa$ $A_n$ from~\cite{Ma73} and its powerset
    automaton that builds a binary tree.  In the \dfa\ on the right
    the transitions of states $\{3\}$ and $\{3,4\}$ are the same as
    for $\{3,5\}$ and $\{3,4,5\}$, respectively.}
  \label{fig:mandl-nfa}
\end{figure}

\begin{theorem}
  For all integers~$n$ and~$\alpha$ such that $\alpha=3\cdot 2^{(n/2)
    - 1} +\beta$ if $n$ is even and $\alpha= 2^{(n+1)/2} +\beta$ if~$n$
 is odd, with $\beta = 2^i -1$ for some integer \mbox{$1\leq i
    \leq \lceil\frac{n-1}2\rceil$,} there exists an $n$-state
  nondeterministic finite automaton accepting a \emph{finite} language
  over a binary alphabet whose equivalent minimal deterministic finite
  automaton has exactly~$\alpha$ states.
\end{theorem}

\begin{proof}
  Let $n,\alpha$ and $\beta$ be as required and $x=n+1-\log(\beta+1)$.
  We construct a minimal automaton~$B_{n,\beta}$ adapting
  $A_{n-1}=(\{1,\dots,n-1\},\{a,b\},\delta_1,1,\{n-1\})$ from above by
  taking a new initial state $0$ and setting the transition function
  $\delta$ to 
  $\delta(0,b)=\{1\}$, $\delta(0,a)=\{1,x\}$, and
  $\delta(q,c)=\delta_1(q,c)$, for $1\leq q\leq n-1$ and letter $c\in\{a,b\}$.
 
  \begin{sloppypar}
  Let $A'_{n-1}$ and $B'_{n,\beta}$ be the powerset automata of
  $A_{n-1}$ and $B_{n,\beta}$.  Then, by reading words~$bw'$ for
  \mbox{$w'\in\{a,b\}^*$}, all states of $A_{n-1}'$ are reachable in
  $B'_{n,\beta}$.  Together with the initial state $\{0\}$, these
  are~$2^{(n-1)/2 +1}$ states if $n$ is odd, and~\mbox{$3\cdot 2^{n/2
      -1}$}
  states if $n$ is even.  For considering words of the form $w=aw'$,
  for $w'\in\{a,b\}^*$, let $k=\lceil\frac{n-1}2\rceil$.  Then
  $k+1\leq x\leq n$ and we reach the states
  \[\delta'(\{0\},aw)=\delta'(\{1,x\},w) = \delta'(\{1\},w) \cup \delta'(\{x\},w).\]
  These states differ from the ones in $A'_{n-1}$ as long as
  $\delta'(\{x\},w)\neq\emptyset$, and this holds if and only if 
	\mbox{$|w|\leq n-x$}. There are $2^{n-x+1}-1=\beta$ such words, so there are
  $\beta$ additional states.
  \end{sloppypar}
\end{proof}


\end{document}